%% file: main.tex
\documentclass[letterpaper, 10pt, conference]{ieeeconf}
\IEEEoverridecommandlockouts                  
\input{macros}

\title{Expressive Power of WSTL Formulas for Learning to Rank}

\author{Ruya Karagulle, Gustavo A. Cardona, Necmiye Ozay, Cristian-Ioan Vasile
\thanks{R. Karagulle and N. Ozay are with the University of Michigan. C.I. Vasile is with Lehigh University, G.A. Cardona is with Virginia Commonwealth University. Corresponding author email: ruyakrgl@umich.edu}\thanks{This work is supported in part by NSF Grant  TI-2303564 and IIS-2442644.
}}

\begin{document}
\maketitle

\begin{abstract}
Weighted Signal Temporal Logic (WSTL) is increasingly used as a scoring function in learning-to-rank problems of trajectories with safety guarantees, where its weighted quantitative semantics serve as a parametrized utility function. Despite the growing interest, prior work only assumes the expressiveness of WSTL formulas and empirically demonstrates its utility in capturing diverse preferences. This work focuses on the correctness of this assumption and asks whether using WSTL formulas as scoring functions is theoretically justified. We formalize two concepts: first, rank-realizability, which asks whether all rankings of a given signal set are achievable by varying weights, and, second, rank-capacity, the maximum signal set size for which the formula is rank-realizable. We propose a Mixed-Integer Linear Program to decide rank-realizability, and derive constructive lower bounds for rank-capacity. Experiments on a robotic navigation task show that while a practical WSTL specification may fail to be rank-realizable on a set with similar trajectories, its rank-capacity exceeds the size of the trajectory set. Analysis of Boolean-equivalent formulas reveals that formula structure affects expressivity and that rank-capacity can be increased without altering qualitative semantics.
\end{abstract}

\section{Introduction}
Temporal logics are powerful tools for integrating formal guarantees into machine learning applications, including reward shaping in Reinforcement Learning \cite{Li2017rltl, Balakrishnan2019rewardstl, jackermeier2025deepltl}, learning temporal logic formulas from labeled trajectories to enable integration with downstream control synthesis methods \cite{Karagulle2022stlclassification, BoVaPeBe-HSCC-2016}, and constructing neuro-symbolic architectures that ground neural networks in a logical framework \cite{riegel2020logical, badreddine2022logic}. In particular, Weighted Signal Temporal Logic (WSTL) has emerged as a tool for encoding preferences and priorities over system specifications by assigning weights to operators \cite{Mehdipour2021wstl,Yan2021neural}. This weighted structure naturally casts WSTL semantics as a parametrized utility function and motivates its use in learning from human feedback with safety guarantees \cite{Karagulle2024spl, Karagulle2024sapl, Karagulle2024cccc}. Prior work assumes that WSTL formulas, and their quantitative semantics, can serve as a score that orders trajectories by preference, and validates this empirically for tasks such as controller synthesis and user satisfaction. However, a fundamental question remains unanswered: Is WSTL theoretically expressive enough to serve this role? That is, for a given formula with adjustable weights, $(i)$ can it induce any desired ranking over a set of signals, or are there orderings it is structurally incapable of producing, regardless of how weights are chosen; $(ii)$ are there any fundamental limitations to the maximum number of signals for which the formula can induce all possible rankings?

This paper takes a step back from prior work and asks whether the assumption of using WSTL as a scoring function is theoretically justified, and under what conditions on the formula. We adopt a learning-theoretic perspective on WSTL
and formalize two concepts: $(i)$ rank-realizability: for a given finite set of signals, can a formula with adjustable weights realize all possible orderings by varying its weights, $(ii)$ rank-capacity: how large a signal set can be while still admitting rank-realizability. These questions are parallel to classical learnability results in the classification problem, such as VC-dimension, but are focused specifically on the use of WSTL formulas in learning-to-rank problems.

For rank-realizability, we analyze how weights propagate through the formula to the predicates at time instances, and whether each signal's robustness can be independently steered by adjusting distinct weight valuations. Building on this, we provide sufficient conditions for rank-realizability and propose a Mixed-Integer Linear Program to decide whether a formula is rank-realizable on a given finite set of signals. Moreover, we derive constructive lower bounds on rank-capacity by treating signals as decision variables as well. We demonstrate both concepts on a robot navigation task and a family of equivalent formulas in Boolean semantics, showing empirically that the formula structure, not merely the parameter count, affects the expressivity.

\section{Preliminaries}
We consider signals $\signal: \timedomain \to \sdomain$, where $\timedomain \subseteq \integers_{\geq0}$ is the discrete time domain and $\sdomain \subseteq \reals^{m}$ is the $m$-dimensional real-valued signal domain. STL is used to reason about signals, and a well-formed STL formula is defined by the grammar $\phi ::= \top \mid \predicate \mid \lnot \phi \mid \bigwedge_{\ell\in\range{1}{n}} \phi_\ell \mid \phi_1 \U_{[\lb,\ub]} \phi_2$, 
where $\top$ is the Boolean True, $\predicate$ is a predicate of the form $\predicate(\signal(t)):= \big(h_\predicate(\signal(t)) \geq 0 \big)$ where $h_\predicate: \sdomain \to \reals$ is a continuous function that maps the signal value at time $t$ to a real value. The operator $\lnot$ is the negation, $\land$ is the conjunction, with $n$ being the number of subformulas under the conjunction operator and $\range{1}{n}$ denoting $\{1, 2, \ldots, n\}$. Finally, $\U_{[\lb,\ub]}$ is the ``Until" operator
\footnote{Additional operators; disjunction $\lor$, Always $\G_{[\lb,\ub]}$, and Eventually $\F_{[\lb,\ub]}$ can be derived from the grammar as $\bigvee_{\ell \in \range{1}{n}} \phi_\ell= \lnot\bigwedge_{\ell\in\range{1}{n}} \lnot \phi_\ell$, $\F_{[\lb,\ub]}\phi = \top \U_{[\lb,\ub]} \phi$, and $\G_{[\lb,\ub]}\phi = \lnot(\F_{[\lb,\ub]}\lnot \phi)$.}. 
Subscript ${[\lb, \ub]}$ with $\lb, \ub \in \integers_{\geq0}$ and $\lb \leq \ub$, defines the time interval in $\timedomain$ that the temporal operator is acting on.
When the time interval includes the whole signal length, we omit the time interval subscript. We refer to \cite{Donze2010} for qualitative and quantitative semantics, also called \emph{robustness}. If a signal $\signal$ satisfies (violates) a formula $\phi$ at time $t$, it is shown as $(\signal, t)\models \phi$ ($(\signal, t) \not \models \phi$). The robustness of an STL formula $\phi$ over a signal $\signal$ at time $t$ is denoted as $\rob(\signal, \phi, t)$. When $t=0$, we denote it as $\rob(\signal, \phi)$. 
For finite signals, we assume that the signal length spans the time horizon of the formula, which is a function of the temporal operators' time intervals.
An STL formula is in positive normal form (PNF) when negation only appears in front of predicates. Any STL formula can be written in PNF, and, throughout this paper, we consider STL formulas in PNF.

WSTL extends STL by assigning strictly positive weights to operators, reflecting the relative importance of subformulas and/or time instances \cite{Mehdipour2021wstl}. Its syntax is defined as
$\phi ::= \top \mid \predicate \mid \lnot \phi \mid \phi_{1} \land^{\weight} \phi_{2} \mid \phi_{1} \U_{[\lb,\ub]}^{\weight^1,\weight^2} \phi_{2}$,
where $\weight \in \reals_{>0}^2$ and $\weight^1, \weight^2 : \range{\lb}{\ub} \to \reals_{>0}$. 
All operators are interpreted as in STL. 
The quantitative semantics of WSTL is called \textit{the \wstlrob}~\cite{Mehdipour2021wstl,Karagulle2024spl} and recursively defined as:
\begin{equation}\label{eq:quantsemantics}\arraycolsep=1pt
\begin{array}{rcl}
\wrob(\signal, \top, t) &=& \infty, \\
\wrob(\signal, \predicate, t) &=& \rob(\signal,\predicate,t), \\
\wrob(\signal, \lnot \phi, t) &=& -\wrob(\signal,\phi, t), \\
\wrob(\signal, \land^{\weight}_{\ell \in \range{1}{n}} \phi_\ell, t) &=& \min\limits_{\ell \in \range{1}{n}} \big (\weight_\ell \ \wrob(\signal, \phi_{\ell},t) \big ), \\
\wrob(\signal, \phi_1 \U_{[\lb,\ub]}^{\weight^1,\weight^2} \phi_2, t)  &=&  \max\limits_{t' \in \range{\lb}{\ub}} \hspace{-0.1cm}\big( \min ( \weight^1_{t'}\ \hspace{-0.25cm} \min\limits_{t'' \in \range{t}{t+t'}} \hspace{-0.25cm} \wrob(\signal,\phi_1,t''), \\ && \weight^2_{t'} \ \wrob(\signal, \phi_2, t+t') ) \big ).
\end{array}
\end{equation}
The derived operators have the following definitions:
\begin{equation}\label{eq:quantsemanticsdr}\arraycolsep=1pt
\begin{array}{rcl}
\wrob(\signal, \lor^{\weight}_{\ell \in \range{1}{n}} \phi_\ell, t) &=& \max\limits_{\ell \in \range{1}{n}} \big (\weight_\ell \ \wrob(\signal, \phi_\ell,t)\big ), \\
\wrob(\signal, \G^{\weight}_{[\lb,\ub]}\phi, t)  &=&  \min\limits_{t' \in \range{\lb}{\ub}} \big(w_{t'} \ \wrob  (\signal,\phi, t + t')\big), \\
\wrob(\signal, \F^{\weight}_{[\lb,\ub]}\phi, t)  &=&  \max\limits_{t' \in \range{\lb}{\ub}} \big(w_{t'} \ \wrob  (\signal,\phi, t + t')\big).
\end{array}
\end{equation}

We denote \wstlrob at time $t=0$ using $\wrob(\signal, \phi)$. WSTL quantitative semantics recovers STL's quantitative semantics when all weights are equal to one. It is shown that the quantitative semantics in \eqref{eq:quantsemantics} and \eqref{eq:quantsemanticsdr} are sound with respect to the quantitative semantics of STL \cite{Karagulle2024spl, Mehdipour2021wstl}. 

To enable learning, we treat weights as tunable parameters and define an extension to WSTL called Parametric Weighted Signal Temporal Logic (PWSTL) \cite{Karagulle2022stlclassification}, denoting the set of unknown parameters as $\weightset$, and PWSTL formulas as $\phi_{\weightset}$. The weight set is the set (or a subset) of all possible weights in the formula. We consider $\weightset$ as the set of all possible weights, unless otherwise defined. A valuation $\wval$ instantiates a WSTL formula $\phi_{\wval}$. 

\section{Expressivity of WSTL for Learning to Rank} \label{sec:expressivity-ltr}
In machine learning, learning to rank is one of the core problems of preference learning~\cite{Cao2007ltr}. Given a set of items $X = \{x_1, \ldots, x_d\}$, the learning-to-rank problem searches for a utility function $f: X \to \reals$, such that for a given ranking of items, e.g., $x_1 \succ \ldots \succ x_d$, the induced utility values preserve the same order, i.e., $f(x_1) > \ldots > f(x_d)$. We use ``ranking'' for a strict comparison of items (no ties), and ``ordering'' for comparisons that allow ties.

In this paper, we focus on the learning-to-rank problem of signals. For this task, we consider the \wstlrob of a WSTL formula template as a parametric utility function over signals. Formally, given a PWSTL formula $\phi_{\weightset}$ and $d$ signals $\set{X} = \{\signal_1,\ldots,\signal_d\}$, 
define the utility function $\boldsymbol{\rho}_\phi(\wval) = \mat{\wrob(\signal_1,\phi_{\wval}) \, \cdots \, \wrob(\signal_d,\phi_{\wval})}^T$
that maps each valuation $\wval$ of $\weightset$ to a vector of robustness values. Each valuation defines an ordering over $d$ signals by sorting the entries of $\boldsymbol{\rho}_\phi(\wval)$, achieving a set of orderings by adjusting weights. We focus on rankings throughout the paper.

Formally, a ranking $\perm: \range{1}{d} \to \range{1}{d}$ is a rearrangement of $d$ signals representing $\signal_{\perm(1)}\succ \ldots \succ \signal_{\perm(d)}$.
Let $\permspace^d = \{\perm: \range{1}{d} \to \range{1}{d} \mid \perm \text{ is a bijection} \}$ denote the set of all rankings of $d$ signals. The \emph{expressivity} of $\phi$ over a finite set $\set{X}$ in a learning-to-rank setting is captured by how large the set of rankings that are achievable by $\boldsymbol{\rho}_\phi(\wval)$ is:

\begin{definition}[Expressivity]
\label{defn:ranking-set}
Let $\set{X}$ be a set of $d$ signals and $\phi_\weightset$ a PWSTL formula with parameter set $\weightset=\mathbb{R}^p_{>0}$, where $p$ is the number of parameters. The expressivity of $\phi$ over $\set{X}$ is its ranking set with respect to $\set{X}$, i.e.,
\begin{equation}\label{eq:ranking-set}\begin{split}
    \langle \phi \rangle_{\set{X}} =  \{\perm \in \permspace^d \mid & \exists\,\wval \in \mathbb{R}^p_{>0}, \\ & \wrob(\signal_{\perm(1)}, \phi_{\wval}) > \cdots >  \wrob(\signal_{\perm(d)}, \phi_{\wval}) \}.\end{split}
\end{equation}
We say $\phi$ is \emph{rank-realizable} over $\set{X}$ if $|\langle \phi \rangle_{\set{X}}| = |\set{X}|!$.
\end{definition}
Let us illustrate these concepts with an example. 
\begin{example}\label{example:1}
    Let $\phi = (x \ge 0) \land (y\ge 0)$ be a formula. Recall the \wstlrob computation $\wrob(\signal, \phi_{\wval}) = \min(w_1 x(0), w_2y(0))$. Let $\signal_1 = [x_1(0) = 8, y_1(0) = 1/8]$, $\signal_2 = [x_2(0)=2, y_2(0) =1/2]$, $\signal_3 = [x_3(0) = 1/2, y_3(0) = 2]$, and $\signal_4 = [x_4(0)= 1/8, y_4(0)= 8]$ be 4 signals. The ranking set of $\phi$ with respect to $\set{X} = \{\signal_1, \signal_2, \signal_3, \signal_4\}$ is: 
    \begin{align*}
    \langle \phi \rangle_{\mathcal{X}} = \{ &(1,2,3,4), (2,1,3,4), (2,3,1,4), (3,2,1,4), \\ &(2,3,4,1), (3,4,2,1), (4,3,2,1) \},
    \end{align*} yielding $7$ distinct rankings out of $24$ possible ways. Now, let $\tilde{\phi}=\phi\lor\phi$ be another WSTL formula with $\wrob(\signal, \tilde{\phi}_{\wval}) = \max\big(w_1\min(w_3 x(0), w_4y(0)),w_2\min(w_5x(0), w_6y(0))\big)$. Formula $\tilde{\phi}$ can achieve all rankings for the same signal set, meaning $\tilde{\phi}$ is rank-realizable with respect to $\set{X}$. For example, set $\wval_1=\wval_2=1$, $\wval_3= 8\delta_1$, $\wval_4= 1/2\delta_2$, $\wval_5 = 1/2\delta_3$, and $\wval_6 = 8\delta_4$, where $\delta_i \in (1/8,1)$. Any permutation of the $\delta_i$ values corresponds to a permutation of the induced ranking over the signals.
\end{example}

Example~\ref{example:1} illustrates how two formulas can differ in expressivity on the same set. While $\tilde{\phi} = \phi \lor \phi$ is rank-realizable over $\set{X}$, naively adding more weights via $\hat{\phi} =  \phi \land \phi$ does not help, as the outer conjunction cannot increase expressivity beyond $\phi$ itself. That is why no four signals exist for which $\hat{\phi}$ is rank-realizable. This shows that expressivity is not monotone in the number of weights, and two formulas with the same meaning in the Boolean semantics may have different ranking sets and expressive powers.

Computing $\langle \phi \rangle_\mathcal{X}$ is a challenging problem, as naively counting the realizable rankings is intractable due to the factorial complexity. Instead, we study a simpler yet fundamental question: is $\phi$ rank-realizable over $\set{X}$? Answering this question reveals whether a formula can fully distinguish preferences among the given signals, which is a nontrivial problem since rank-realizability has no simple relation to the number of parameters or to the formula's structure.

Additionally, we are interested in the fundamental limitations on a formula's expressivity over the entire signal domain, that is, the maximum number of signals for which the formula is rank-realizable. We call this property \emph{rank-capacity} of a formula. Rank-capacity is conceptually related to \emph{VC-dimension} in machine learning theory. That said, results from VC-dimension analysis cannot be extended to rank-capacity, as there is no simple way to recast it in terms of shattering of items. We define rank-capacity as follows:
\begin{definition}[Rank-capacity]\label{def:rank-capacity}
    Let $\phi_\weightset$ be a PWSTL formula. The rank-capacity of $\phi$ is the cardinality of the largest signal set $\set{X}$, such that $\phi$ is rank-realizable over $\set{X}$, that is,
    \begin{equation}
        \textrm{RC}(\phi) = \max |\set{X}| \ \suchthat \ |\langle \phi\rangle_{\set{X}}| = |\set{X}|!.
    \end{equation}
\end{definition}

In the next two sections, we provide a sufficient criterion and a mixed-integer linear programming (MILP) approach to check rank-realizability for a given PWSTL formula and set of signals, and extend our findings to rank-capacity to determine lower bounds for $\textrm{RC}(\phi)$. The next two sections develop the theoretical foundations to answer both problems. Fig.~\ref{fig:summary-chart} provides an overview of the main results and logical dependencies to help readability.
\input{flow_chart_iterations}

\section{Rank-Realizability of WSTL formulas}

In this section, we are given a set $\set{X}= \{\signal_1, \ldots, \signal_d\}$ of $d$ signals. Formally, we address the following problem:

\begin{problem}[Rank-Realizability Problem]
\label{prob:ranking-set}
Let $\set{X}$ be a set of $d$ signals and $\phi_\weightset$ a PWSTL formula with parameter set $\weightset=\mathbb{R}^p_{>0}$, where $p$ is the number of parameters. Decide if $\phi$ is rank-realizable with respect to $\set{X}$, that is, $|\langle\phi\rangle_{\set{X}}| = d!$.
\end{problem}

A necessary condition for rank-realizability is that all signals in $\set{X}$ share the robustness sign. Otherwise, as WSTL semantics is sound, it is not possible to change a ranking between violating and satisfying signals. Therefore, for the rest of the analysis, we assume $\rob(\signal_i, \phi) > 0$ for all $\signal_i \in \set{X}$. Extension to mixed-sign signal sets is discussed in Sec.~\ref{sec:mixed-sign-signals}.

First, we introduce key concepts to provide a sufficient condition for rank-realizability, namely, Th.~\ref{thm:indep-sufficient} and Pb.~\ref{prob:rank_realizability}. Then, we formulate rank-realizability as a Mixed-Integer Program (MIP) in Sec.~\ref{sec:rank-realizability-base-weights}. We will show in Sec.~\ref{sec:rank-realizability-original-weights} that it is possible to recast it as a linear problem either by treating derived variables as decision variables or introducing a $\log$-transformation defined in~\cite{karagulle2025sopl}.

\paragraph{Ordering Space}
Define the \emph{ordering space of $\perm$}, denoted as $\robspace_\perm \subseteq \reals^d$, as the closed set of all robustness values inducing $\signal_{\perm(1)} \succeq \signal_{\perm(2)} \succeq \cdots \succeq \signal_{\perm(d)}$. Ordering spaces capture possible orderings of signals independent of any particular WSTL formula. Recall the utility function $\boldsymbol{\rho}_{\phi}$  maps weight valuations to an ordering space. Finally, define the center $\mathscr{C} = \{\alpha\mathbf{1}_d \mid \alpha \in \reals\}$ as the set of $d$-dimensional vectors whose entries are all equal to $\alpha$ in $\reals^d$. 



\paragraph{Base weights}
In a WSTL formula $\phi_{\wval}$, the weight valuation $\wval \in \mathbb{R}^p_{>0}$ assigns coefficients to operators. By the positive homogeneity of $\min$ and $\max$, each operator weight can be carried inside its operands, recursively, propagating to the predicates. As a result, each predicate evaluated at a particular time instance, called \emph{predicate-time pair} $(\predicate, t)$, is multiplied by the product of all operator weights along the path from the outermost operator to that predicate. We refer to this product as the \emph{base weight} of  $(\predicate, t)$, which acts on $h_\predicate(\signal(t))$. We note that each appearance of a subformula-time pair is considered distinct, yielding a tree structure representing formula syntax \cite{karagulle2025sopl}. When there is a repetition of subformulas, e.g., $\phi = \phi_1 \land \phi_1$, we rename them with identifiers to make the distinction traceable.

Let $\phi_1 \subf \phi_2$ denote that $\phi_1$ is a subformula of $\phi_2$. If $\phi_1$ is an operand of $\phi_2$, then $\phi_1$ is called a child of $\phi_2$. We write $(\phi_1, t_1) \subf (\phi_2, t_2)$ to denote that the evaluation of $\phi_1$ at $t_1$ is required for evaluating $\phi_2$ at time $t_2$, and define the path from $(\phi, 0)$ to $(\varphi, t)$ as $\fpath_\phi(\varphi, t) = \{(\varphi', t') \mid (\varphi, t) \subf (\varphi', t') \subf (\phi, 0) \}$ (i.e., a path in the aforementioned tree).

\begin{definition}[Base Weights]
\label{def:base-weights}
Let $\phi_{\wval}$ be a WSTL formula with $\wval \in \mathbb{R}^p_{>0}$. For each predicate-time pair $(\predicate, t)$, the \emph{base weight} $\widetilde{w}_{\predicate,t}(\wval) = \prod_{(\varphi, t') \in \fpath_\phi(\predicate,t)} \wval_{\varphi,t'}$ is the multiplication of all weights on its path.
The collection $\widetilde{\wval} = \{\widetilde{w}_{\predicate,t}(\wval)\}$ of all $\tilde{p}$ base weights defines the formula $\phi_{\widetilde{\wval}}$.
\end{definition}

Base weights represent the case where operator-level weights are collapsed into a single multiplication at each predicate-time pair, capturing how weights propagate through the formula to influence each predicate's contribution to the overall robustness. This aggregation allows extending the positive homogeneity of $\min$/$\max$ to \wstlrob.

\begin{proposition}[Base weight homogeneity] \label{prop:base-homogeneity}
Let $\phi_{\widetilde{\wval}}$ be a WSTL formula with base weights $\widetilde{\wval}$.  
Then, for any $\alpha > 0$ and any signal $\signal$, $\wrob(\signal, \phi_{\alpha \widetilde{\wval}})  = \alpha \cdot \wrob(\signal, \phi_{\widetilde{\wval}}).$
\end{proposition}


Proof directly follows from the positive homogeneity of $\min$/$\max$. Prop.~\ref{prop:base-homogeneity} bounds the search space in weight synthesis problems. As long as the product of scalings equals $\alpha$ in all paths from predicate-time pairs to the outermost operator, we can positively scale the original weights arbitrarily. Prop.~\ref{prop:base-homogeneity} provides structure over ordering spaces as well.

\begin{corollary}\label{cor:center-cover}
Given $d$ signals, and a WSTL formula $\phi_{\widetilde{\wval}}$ with base weights $\widetilde{\wval} \in \reals_{>0}^{\tilde{p}}$ that achieves ranking $\perm \in \permspace^d$, then any positive scaling of $\widetilde{\wval}$ achieves the same ranking $\perm$. Formally, $\forall \phi$,  $\forall \perm \in \permspace^d$, $\boldsymbol{\rho}_\phi(\widetilde{\wval}) \in \robspace_\perm$ implies $\forall \alpha>0$ we have $\boldsymbol{\rho}_\phi(\alpha \widetilde{\wval})\in \robspace_\perm$.
\end{corollary}

\begin{proof}
By Prop.~\ref{prop:base-homogeneity}, scaling $\widetilde{\wval}$ by $\alpha>0$ scales all robustness values uniformly, preserving the ordering.
\end{proof}

\paragraph{Critical path}
The \emph{critical path} of a signal is the path from a predicate-time pair to the outermost operator along which that pair determines the optimizer at each operation, making the \wstlrob value at that predicate-time pair equal to the signal's overall \wstlrob value.

\begin{definition}[Critical Path]\label{def:critical-path}
    Given a signal $\signal$, formula $\phi$ and predicate-time pair $(\predicate, t) \subf (\phi, 0)$, the path $\fpath_\phi(\predicate, t)$ (and the predicate-time pair $(\predicate, t)$) is \emph{critical} if $(\predicate, t)$ determines the weighted robustness of the formula and each subformula along its path, that is, $\wrob(\signal,\varphi,t') = \widetilde{w}_{\predicate, t} \cdot h_\predicate(\signal(t))$ for all $(\varphi,t')\in\fpath_\phi(\predicate,t)$.
    If exactly one such pair exists, we call the predicate-time pair and its path \emph{unique critical}.
\end{definition}

The critical predicate-time pair can be defined for any subformula-time pair $(\varphi, t)$ of $\phi$. Let $(\predicate^\dagger, t^\dagger)$ be the critical predicate-time pair of $(\varphi,t)$. We define \emph{the critical value of $(\varphi,t)$} as $ \wrob_c(\signal,\varphi,t) = \widetilde{\wval}_{\predicate^\dagger, t^\dagger}h_{\predicate^\dagger}(s(t^\dagger))$.

We now establish a sufficient condition for rank-realizability from the existence of a center point $\widetilde{w}$ such that $\boldsymbol{\rho}_\phi(\widetilde{\wval}) \in \mathscr{C}$, combined with a local independence property on base weights. By Cor.~\ref{cor:center-cover}, we can reduce the existence condition to checking $\boldsymbol{\rho}_\phi(\widetilde{\wval}) = \mathbf{1}_d$. Let $K(\phi)= \{\widetilde{\wval} \mid \boldsymbol{\rho}_{\phi}(\widetilde{\wval}) = \mathbf{1}_d\}$ denote the set of base weights achieving this, and define the local independence condition as follows.



\begin{definition}[Independent Sensitivity Directions] \label{def:indep-sens-dir}
Let $\set{X}$ be a set of $d$ signals, $\phi_{\weightset}$ a PWSTL formula, and 
$K(\phi)$ non-empty.
The sensitivity directions are locally independent around 
$\boldsymbol{1}_d$
if there exists $\varepsilon \in (0,1)$ such that for all $\boldsymbol{r} \in \reals^d$ with $\|\boldsymbol{r}- \boldsymbol{1}_d\|_\infty \leq \varepsilon$, there exists $\widetilde{\wval}$ such that $\boldsymbol{\rho}_{\phi}\left(\widetilde{\wval}\right) = \boldsymbol{r}$.
\end{definition}


Intuitively, small independent perturbations in the \wstlrob values of signals can be achieved through independent rescalings of base weights $\widetilde{\wval} \in K(\phi)$.
If we start from an all-equal \wstlrob point and perturb the \wstlrob values of signals independently, we can rank them arbitrarily. 

\begin{theorem}
\label{thm:indep-sufficient}
Given $\set{X}$ and a PWSTL formula $\phi_{\weightset}$, 
if the sensitivity directions are locally independent around $\boldsymbol{1}_d$, 
then $\boldsymbol{\rho}_{\phi}$ achieves all rankings.
\end{theorem}

\begin{proof}
Let $\perm \in \permspace^d$ and $\boldsymbol{r}_\perm \in \mathring{\robspace}_\perm$ (where $\mathring{}$ denotes the interior of a set) such that $\|\boldsymbol{r}_\perm-\boldsymbol{1}_d\|_\infty < \varepsilon$.
Since sensitivity directions are independent around 
$\boldsymbol{1}_d$
there exists $\widetilde{\wval}_\perm$ achieving $\boldsymbol{\rho}_\phi(\widetilde{\wval}_\perm) = r_\perm$.
As $\perm$ and $\boldsymbol{r}_\perm$ are arbitrary, it follows that $\boldsymbol{\rho}_{\phi}$ can achieve all rankings.
\end{proof}

A sufficient practical condition for local independence around 
$\boldsymbol{1}_d$
is the existence of $\widetilde{\wval}^\star\in K(\phi)$ such that each signal's \wstlrob is determined by a distinct and unique critical path.  That is, for each signal, we have $\wrob(\signal_i, \phi_{\widetilde{\wval}^{\star}}) = \widetilde{\wval}^{\star}_{\predicate_{i}^\dagger, t_i^\dagger}h_{\predicate_{i}^\dagger}(\signal(t_i^\dagger))$, where $(\predicate_i^\dagger, t_i^\dagger)$ pairs are unique critical predicate-time pairs, and each signal has a distinct critical predicate-time pair, then each signal's \wstlrob can be independently rescaled by
substituting $\widetilde{\wval}_{\predicate_i^\dagger, t_i^\dagger}$ in place of $\widetilde{\wval}^{\dagger}_{\predicate_i^\dagger, t_i^\dagger}$.
A small scaling for each $\widetilde{\wval}^{\star}_{\predicate_{i}^\dagger, t_i^\dagger}h_{\predicate_{i}^\dagger}$ will not change the critical path but will perturb $\wrob(\signal_i, \phi_{\widetilde{w}})$ by at most $\varepsilon$. This ensures local control over rankings, and we aim to find such base weights:
\begin{problem}\label{prob:rank_realizability}
    Given $\set{X}$, and $\phi$, find $\widetilde{\wval}^{\star} \in \reals_{>0}^{\tilde{p}}$ such that 
    \begin{enumerate}
        \item its robustness vector is in $\mathscr{C}$, i.e., $\boldsymbol{\rho}_\phi(\widetilde{\wval}^{\star}) = \mathbf{1}_d$,
        \item for each signal $\signal_i \in \set{X}$, its robustness value is determined by a unique critical path,
        \item each signal's critical path is distinct, and
        \item base weights $\widetilde{\wval}^{\star}$ are $\widetilde{w}_{\predicate,t}(\wval^{\star}) = \hspace{-0.2em}\prod\limits_{(\varphi, t') \in \fpath_\phi(\predicate,t)} \hspace{-0.2em}\wval^{\star}_{\varphi,t'}$.
    \end{enumerate}
\end{problem}

We formalize the feasibility of Pb.~\ref{prob:rank_realizability} as a sufficient condition for rank-realizability.

\begin{theorem}[PWSTL Rank-Realizability]\label{thm:rank-real-equiv}
    If Pb.~\ref{prob:rank_realizability} is feasible, then PWSTL formula $\phi$ is rank-realizable.
\end{theorem}

\begin{proof}[Sketch]
We show that conditions in Pb.~\ref{prob:rank_realizability} imply independent sensitivity directions around $\widetilde{\wval}^\star$. Cond. 1 ensures $\widetilde{\wval}^\star \in K(\phi)$, and $\wrob(\signal_i, \phi_{\widetilde{\wval}^{\star}})=1$ for all signals. Since each signal's critical path is unique (Cond. 2), at each subformula along the unique critical path, there is a unique optimizer in the recursions~\eqref{eq:quantsemantics}-\eqref{eq:quantsemanticsdr}. Otherwise, there would exist another critical path, contradicting uniqueness.

Let $(\predicate^\dagger_i, t^\dagger_i)$ denote the critical predicate-time pair of signal $s_i$. 
Let $\varepsilon_i = \min_{\varphi,t} |\wrob_c(\signal,\varphi,t) - 1|$, where the minimum is over the critical values of non-critical children of subformulas along $\fpath_\phi(\mu^\dagger_i, t^\dagger_i)$.
We have $\varepsilon_i > 0$ following from the uniqueness of the critical path. 
Define $\varepsilon = \min(\min_i(\varepsilon_i), 0.5)$ where bound $0.5$ ensures $\varepsilon$-ball lies in the positive quadrant. Take any $\|\boldsymbol{r}-\boldsymbol{1}_d\|_\infty \leq \varepsilon$. Since $|r_i - 1| \leq \varepsilon_i$, the critical 
path for $\signal_i$ is preserved under any weight perturbation 
achieving $\wrob(\signal_i, \phi_{\widetilde{\wval}}) = r_i$. By Cond.~3, each $r_i$ depends on a distinct base weight $\widetilde{\wval}_{\predicate_i^{\dagger}, t_i^{\dagger}}$, and $\widetilde{\wval}_{\predicate_i^{\dagger}, t_i^{\dagger}}$ can be adjusted independently to achieve $r_i$, guaranteeing existence of $\widetilde{\wval}$ with $\boldsymbol{\rho}_\phi(\widetilde{\wval}) = \boldsymbol{r}$, making the sensitivity directions locally independent around $\boldsymbol{1}_d$. By Thm.~\ref{thm:indep-sufficient}, $\boldsymbol{\rho}_\phi$ achieves all rankings. Finally, Cond. 4 ensures the existence of weights $\wval \in \reals^{p}_{>0}$ realizing $\widetilde{\wval}^\star$, required by Defn.~\ref{def:base-weights}.
\end{proof}

We now formulate Pb.~\ref{prob:rank_realizability} as an MIP, then linearize it.

\subsection{An MIP Encoding for Rank-Realizability} \label{sec:rank-realizability-base-weights}
Building on Pb.~\ref{prob:rank_realizability}, we seek a weight valuation that induces unit \wstlrob values for all signals, where each signal's robustness is determined by a unique predicate-time instance. If such a configuration exists with a strictly positive margin $\varepsilon$, small independent rescalings of the base weights can achieve any desired ranking, providing a sufficient certificate of rank-realizability. We formulate this as an MIP using binary variables and the big-M method.

\emph{Decision variables:} Let $\xi^i_{\phi, t}, \eta^i_{\phi, t} \in \mathbb{B}$ encode upper and lower bounds on $\wrob(\signal_i,\phi,t)$, respectively. That is, $\xi^i_{\phi, t}=1$ implies $\wrob(\signal_i, \phi, t) \leq 1$, otherwise unconstrained. Similarly, $\eta^i_{\phi, t}=1$ if $\wrob(\signal_i, \phi, t) \geq 1$, otherwise unconstrained. Let $\widetilde{\wval}_{\predicate, t} \in \mathbb{R}_{>0}$ be the base weight for predicate $\predicate$ at time $t$, shared across all signals. Let $\epsilon^i_{\predicate, t} \geq 0$ be a slack for $(\predicate, t)$, for each $\signal_i$. Let $\zeta^i_{\predicate,t}\in \mathbb{B}$ indicate whether $(\predicate, t)$ is critical for $\signal_i$, and $\varepsilon \geq 0$ be the global slack margin, to be maximized. 

\emph{Predicate Constraints} ($\phi = \predicate$) are encoded as follows: 
\begin{align}
    \widetilde{w}_{\predicate, t} h_\predicate(\signal_i(t)) + \epsilon^i_{\predicate, t} -1 &\leq M (1 - \xi^i_{\predicate, t}) \label{eq:milp-pred-upper-bound},\\
    \widetilde{w}_{\predicate, t}h_\predicate(\signal_i(t)) - \epsilon^i_{\predicate, t} -1&\geq -M (1 - \eta^i_{\predicate, t}) \label{eq:milp-pred-lower-bound},\\
    \zeta^i_{\predicate, t} & \le  \xi^i_{\predicate, t}, \label{eq:milp-flag1} \\
     \zeta^i_{\predicate, t} & \le \eta^i_{\predicate, t}, \label{eq:milp-flag2}\\
     \zeta^i_{\predicate, t} & \ge \xi^i_{\predicate, t}+\eta^i_{\predicate, t} -1, \label{eq:milp-flag3}\\
    \epsilon^i_{\predicate, t} + M \zeta^i_{\predicate, t} &\geq \varepsilon \label{eq:milp-pred-margin}.
\end{align}

Constraints \eqref{eq:milp-pred-upper-bound}-\eqref{eq:milp-pred-lower-bound} gate each weighted predicate value relative to the center via big-M. When both $\xi^i_{\predicate, t}$ and $\eta^i_{\predicate, t}$ are set to one, it enforces $\wrob(\signal_i, \predicate, t) =1$. Constraints \eqref{eq:milp-flag1}-\eqref{eq:milp-flag3} set $\zeta^i_{\predicate,t}=1$ if and only if both bounds are tight, marking $(\predicate,t)$ as critical for $\signal_i$. Constraint \eqref{eq:milp-pred-margin} enforces a uniform margin $\varepsilon$ at all non-critical pairs.

\emph{Conjunction constraints} ($\phi=\land^{\weight}_{\ell\in \range{1}{n}} \phi_\ell$) are as follows:
\begin{align}
    \eta^i_{\phi, t} &= \eta^i_{\phi_\ell, t}, \forall \ell \in \range{1}{n} \quad \text{(Copy-Op)}, \label{eq:copying-conj}\\
    \sum_{\ell=1}^{n} \xi^i_{\phi_\ell, t} &\geq \xi^i_{\phi, t},\label{eq:ub-child-conj1}\\
    \sum_{\ell=1}^{n} \xi^i_{\phi_\ell, t} &\leq \xi^i_{\phi, t} + (n-1) (1 - \eta^i_{\phi, t}),\\
    \sum_{\ell=1}^{n} \xi^i_{\phi_\ell, t} &\leq n \cdot \xi^i_{\phi, t} + \eta^i_{\phi, t}\label{eq:ub-child-conj3}.
\end{align}

The Copy-Op propagates the lower bound to all children. Constraints \eqref{eq:ub-child-conj1}-\eqref{eq:ub-child-conj3} enforce that exactly one child is tight on the upper bound in the centered case, identifying the unique $\min$-determining child.

\emph{Disjunction constraints} ($\phi=\lor^{\weight}_{\ell\in \range{1}{n}} \phi_\ell$) are encoded symmetrically by swapping $\xi \leftrightarrow \eta$ in \eqref{eq:copying-conj}-\eqref{eq:ub-child-conj3}, with Copy-Op applied to the upper bound. The \emph{Always} ($\phi= \G^{\weight}_{[\lb,\ub]}\varphi$) and \emph{Eventually constraints} ($\phi= \F^{\weight}_{[\lb,\ub]}\varphi$) follow the same pattern as conjunction and disjunction respectively, acting over time interval $I=\range{\lb}{\ub}$ rather than subformula indices.

\emph{The \wstlrob value and independence constraints} are encoded as follows
\begin{equation}
    \xi^i_{\phi,0} = \eta^i_{\phi,0} = 1, \; \forall i \in 
    \range{1}{d}, \label{eq:zero-robustness-constraint}
\end{equation}
\begin{equation}
    \sum_{i=1}^d \zeta^i_{\predicate,t} \leq 1, \; \forall 
    \predicate, t. \label{eq:local-independence}
\end{equation}
Constraint \eqref{eq:zero-robustness-constraint} fixes all signals' robustness to the center. Constraint \eqref{eq:local-independence} ensures each predicate-time pair is critical for at most one signal, yielding locally independent weight adjustment directions. After adding base weights conditions 
\begin{equation}\label{eq:baseweights}
    \widetilde{w}_{\predicate,t} = \prod_{(\varphi, t') \in \fpath_\phi(\predicate,t)} w_{\varphi,t'},
\end{equation}
for all base weights,
the final problem is:
\begin{equation}
\label{eq:decision-milp}
    \max_{\xi^i_{\phi,t}, \eta^i_{\phi,t}, \zeta^i_{\phi,t}, 
    \widetilde{w}_{\predicate,t}, \epsilon^i_{\predicate,t}, 
    \varepsilon} \varepsilon \quad \text{s.t.} \quad 
    \eqref{eq:milp-pred-upper-bound}\text{-}
    \eqref{eq:baseweights}.
\end{equation}

\begin{theorem}[MIP Encoding Correctness] \label{thm:milp-correctness}
Given $\set{X}$ and $\phi_\weightset$ in PNF, \eqref{eq:decision-milp} is a sound and complete encoding of Pb.~\ref{prob:rank_realizability}: if feasible with $\varepsilon^*>0$, each signal has a unique critical predicate-time pair distinct from all others, and $\boldsymbol{\rho}_\phi(\widetilde{\wval}^*)=\mathbf{1}_d$. Conversely, any solution to Pb.~\ref{prob:rank_realizability} with strict separation at non-critical pairs yields a feasible \eqref{eq:decision-milp} with $\varepsilon>0$.
\end{theorem}

\begin{proof}[Sketch]
    Soundness follows by structural induction over the formula tree\cite{hopcroft2001introduction, cardona2023flexible}. The base case (predicate) follows directly from \eqref{eq:milp-pred-upper-bound}-\eqref{eq:milp-pred-lower-bound}. When $\xi^i_{\predicate,t}=1$ the model enforces $\widetilde{w}_{\predicate,t} h_\predicate (\signal_i(t))\leq1$. Similarly, $\eta^i_{\predicate,t}=1$ enforces $\widetilde{w}_{\predicate,t} h_\predicate (\signal_i(t))\geq1$; the big-M terms relax the inequality if the indicator is 0. Then, \eqref{eq:milp-flag1}-\eqref{eq:milp-flag3} make $\zeta^i_{\predicate, t}=1$ iff both bounds are active, making $(\predicate,t)$ as the critical pair for $\signal_i$. For the induction step, $\min$-type operators (conjunction, always) propagate lower bounds to all children via Copy-Op while requiring at least one child to be tight on the upper bound; $\max$-type operators do the reverse. Constraint \eqref{eq:zero-robustness-constraint} initiates tightness at the root, which propagates to exactly one critical predicate-time pair per signal. Constraint \eqref{eq:local-independence} guarantees these pairs are distinct, and \eqref{eq:milp-pred-margin} ensures non-critical pairs are separated by $\varepsilon^*>0$.
    Completeness follows by construction: given critical predicate-time pairs $(\predicate_i,t_i)$, set $\zeta^i_{\predicate_i,t_i}=1$, propagate indicators upward, and choose slacks as absolute deviations from one.
\end{proof}

\subsection{MILP encoding in the Original Weight Space} \label{sec:rank-realizability-original-weights}

If we treat base weights $\widetilde{w}_{\predicate,t}$ as independent decision variables by removing Eq.~\eqref{eq:baseweights}, meaning not as a multiplication of weights but as a single variable, \eqref{eq:decision-milp} is already an MILP. However, a feasible solution in the base-weight space does not always correspond to a feasible valuation in the original weight space, since base weights must satisfy multiplicative relations across shared operator weights. Enforcing these dependencies directly leads to multi-linear constraints, turning the problem into an MINLP.

\begin{example}Take $\phi_1 = \G_{[0,1]}\phi$ of Ex.~\ref{example:1} with base weights $\widetilde{w}_{x\ge0,0} =w^{\G}_{0}w_{x\ge0}^{\land}$, $\widetilde{w}_{y\ge0,0} = w^{\G}_{0}w_{y\ge0}^{\land}$, $\widetilde{w}_{x\ge0,1} = w^{\G}_{1}w_{x\ge0}^{\land}$, $\widetilde{w}_{y\ge0,1} = w^{\G}_{1}w_{y\ge0}^{\land}$. The base-weight solution $\widetilde{w}_{x\ge0,0}=\widetilde{w}_{y\ge0,0}=\widetilde{w}_{x\ge0,1}=10$, $\widetilde{w}_{y\ge0,1}=12.5$ may satisfy \eqref{eq:decision-milp}, yet these values are not realizable as products of any $\wval \in \reals_{>0}^4$ since the ratio $\widetilde{w}_{y\ge0,1}/\widetilde{w}_{x\ge0,1} \neq \widetilde{w}_{y\ge0,0}/\widetilde{w}_{x\ge0,0}$ violates Cond. 4 of Pb.~\ref{prob:rank_realizability}.
\end{example}

To recover an MILP in the original weight space, we follow the 
structural pruning and $\log$-transform approach of \cite{karagulle2025sopl}. Structural pruning removes paths whose \wstlrob values do not share the same sign with the overall robustness, thus cannot affect it (e.g., a negative value in a $\max$ dominated by positive operands), leaving only paths that share the same robustness sign with the overall robustness. The $\log$-transform then converts multiplicative weight relations into additions, which is valid thanks to the strictly positive definition of weights and the strictly positive \wstlrob assumption on all signals in $\set{X}$. In particular, we begin by applying structural pruning to signals. We then replace \eqref{eq:milp-pred-upper-bound}--\eqref{eq:milp-pred-lower-bound} by their $\log$-domain counterparts: 
\begin{align}
    \sum_{(\varphi, t') \in \fpath_\phi(\predicate,t)} \log(w_{\varphi',t}) + \log(h_\predicate(\signal_i)) + \epsilon^i_{\predicate,t} &\le M(1-\xi^i_{\predicate,t}),\label{eq:log-pred-upper-bound}\\
    \sum_{(\varphi, t') \in \fpath_\phi(\predicate,t)} \log(w_{\varphi',t}) + \log(h_\predicate(\signal_i)) - \epsilon^i_{\predicate,t} &\ge -M(1-\eta^i_{\predicate,t})\label{eq:log-pred-lower-bound}.
\end{align}

Substituting $v_{\varphi',t} = \log(w_{\varphi',t})$ linearizes the weight terms, and $\log(h_\predicate(\signal_i))$ is precomputed from the given signals, yielding a full MILP in the original weight space. Finally, we set $\log \wrob(\signal_i, \phi_{\wval}) = \log(1)=0$.

\begin{theorem}\label{thm:milp-equiv}
    Given $d$ signals and $\phi$, if maximizer $\varepsilon^*$ of the $\log$-transform of ~\eqref{eq:decision-milp} (after structural pruning) is strictly positive, then $\phi$ can achieve all rankings. 
\end{theorem}
\begin{proof}[Sketch]
    We divide the proof into two parts. $(i)$ Sufficiency of $\varepsilon^*>0$ in the original problem follows from Thm.~\ref{thm:milp-correctness}. $(ii)$  Under the $\log$-transform, multiplicative weight relations become additive, and since $\min$/$\max$ are translation-invariant, the local independence structure of Defn.~\ref{def:indep-sens-dir} is preserved. Thus, $\varepsilon^*>0$ in the $\log$-domain implies rank-realizability in the original domain.
\end{proof}

\subsection{Signal sets with mixed robustness signs}
\label{sec:mixed-sign-signals}
In the previous section, we assumed that  $\set{X}$ contains only signals with positive robustness, $\wrob(\signal, \phi_\weight) >0$ for all $\signal \in \set{X}$ and for all $\weight \in \reals_{>0}^p$ since weights cannot change the sign of the robustness. In case $\set{X}$ contains signals with positive, negative and/or zero robustness, we have the following immediate consequences regardless of any weight choice: (a) all signals with positive robustness are ranked higher than all non-positive robustness signals (zero and negative); (b) all negative robustness signals are ranked lower than all non-negative robustness signals (zero and positive); and (c) all zero-robustness signals are ranked the same and cannot be distinguished via weights. Thus, we partition the signal set $\set{X} = \set{X}_{>0} \cup \set{X}_{=0} \cup \set{X}_{<0}$, where $\set{X}_{\sim 0} = \{\signal \in \set{X} \mid \wrob(\signal, \phi_\weightset) \sim 0 \}$ with $\sim \in \{>, <, =\}$. The best outcome we can hope for is to achieve all possible rankings of the positive and negative robustness signals separately. In this case, while keeping the unique and distinct critical path requirement for all signals in $\set{X}$, we check whether we can achieve $|\set{X}_{>0}|! \cdot |\set{X}_{<0}|! < |\set{X}|!$ rankings. Algorithmically, (a) we solve \eqref{eq:decision-milp} for $\set{X}_{>0}$ and $\phi_\weightset$, and (b) we solve \eqref{eq:decision-milp} for $\set{X}_{<0}$ and $\neg\phi_\weightset$ after converting it to PNF\footnote{Note that taking the PNF does not change the number of weights since none of the negation equalities introduce new weights.}. Note that unique and distinct critical paths allow us to perturb negative and positive robustness values independently. If both problems are feasible and result in positive global slack values, then we obtain the bound $|\set{X}_{>0}|! \cdot |\set{X}_{<0}|!$.

\section{Rank-capacity of WSTL formulas}
The rank-capacity problem concerns the expressivity of a formula across the entire signal domain, rather than a fixed finite set. We follow a computational approach to obtain lower bounds on $\textrm{RC}(\phi)$.

A natural extension for the rank-capacity is treating the signals in the rank-realizability formulation as decision variables as well. Formally, we have:
\begin{equation}
    \label{eq:decision-milp-capacity}
    \max_{\signal_1,\ldots,\signal_d, \widetilde{w}_{\phi, t}, \xi^i_{\phi, t}, \eta^i_{\phi, t}, \zeta^i_{\phi, t}, \epsilon^{i}_{\predicate, t}, \varepsilon} \varepsilon \suchthat \eqref{eq:milp-pred-upper-bound}-\eqref{eq:local-independence}.
\end{equation}

\begin{theorem}\label{thm:capacity-milp}
    If~\eqref{eq:decision-milp-capacity} is feasible with $\varepsilon^*>0$, then $\textrm{RC}(\phi)\ge d$. 
\end{theorem}

Proof follows the fact that the feasibility of ~\eqref{eq:decision-milp-capacity} with $\varepsilon^*>0$ is a sufficient condition for rank-realizability of $d$ signals by Thm.~\ref{thm:milp-correctness}. Treating signals as decision variables preserves the nonlinear nature of the problem. To recover an MILP, we take the logarithm of predicate-time values $h_{\predicate}(\signal(t))$ as unknowns instead of signals directly. Two issues arise from this substitution: $(i)$ positivity: $\log(h_{\predicate}(s(t)))$ is well-defined only when $h_{\predicate}(s(t))>0$, which cannot be guaranteed. For a feasible set of signals, even though the overall \wstlrob is positive for all, there might be predicates with negative values. $(ii)$ shared dimensions: if two different predicates depend on the same dimension of the signal, e.g., $\predicate_1: \signal^1 \ge 2$, and $\predicate_2: \signal^1 \le 5$, their predicate-time values e.g., $h_{\predicate_{1}}(s(t))$ and $h_{\predicate_{2}}(s(t))$ cannot be treated as independent variables. We address the issues one by one.

For $(i)$, we begin by assuming the predicates do not have a shared dimension issue, and show that existence of a predicate-value set $\set{X_\predicate} = \{h_{\predicate}(\signal_i(t))\}$ in the entire domain such that the formula is rank-realizable with respect to it implies the existence of a predicate-value set in the positive quadrant. 
This property allows us to use $\log$-transform for signal values as well.

\begin{theorem}\label{thm:predicate-positivity}
    Consider $\phi_{\weightset}$ in PNF with predicates $\predicate:= \signal^j \sim c_{\predicate}$, carrying relations $\sim \in \{\ge, \le\}$, where $\signal^j$ represents the $j^{\text{th}}$ dimension of $\signal$ and $c_{\predicate} \in \reals$ represents the threshold value. Assume each predicate is defined over a distinct signal dimension. If there exists a predicate-value set for $d$ signals in the entire domain, then there exists a predicate-value set for $d$ signals in the positive quadrant. 
\end{theorem}

\begin{proof}[Sketch]
    We treat the two cases to transfer the predicate-value set living in the entire domain to the positive quadrant. \emph{Case 1: all robustness values are strictly positive.}  Any predicate-time value $h_{\predicate}(\signal(t)) < 0$, because the overall robustness is positive $\rob(\signal_i,\phi)>0$, must be dominated inside a $\max$ operator by a positive operand; otherwise, the overall robustness would be negative. Replacing it with an arbitrarily small $\varepsilon>0$ leaves all $\max$ values and subsequent rankings unchanged. \emph{Case 2: All robustness values are strictly negative.} Any predicate-time value $h_{\predicate}(\signal(t)) > 0$ must be dominated inside a $\min$ operator by a negative operand; otherwise, the overall robustness would be positive. Therefore, replacing positive values with arbitrarily large numbers does not affect any $\min$. For predicate-time values $h_\predicate(s(t))<0$, replacing them with $-1/h_\predicate(s(t))$ preserves their relative ordering.
\end{proof}

Thm.~\ref{thm:predicate-positivity} justifies restricting the search 
to the positive quadrant, where the $\log$-transform is always defined.

For $(ii)$, when predicates share the same signal dimension, we use a surrogate encoding $\log(\signal_j)-\log(c_\predicate)$ in place of  $\log(\signal^j - c_\predicate)$. This preserves satisfaction due to $\log$'s monotonicity (whenever $\signal^j - c_\predicate > 0$, we have $\log(\signal^j)-\log(c_\predicate)>0$) but does not exactly preserve the zero level set of the original robustness. Therefore, candidate signals found via the surrogate must be certified by plugging them into the rank-realizability problem~\eqref{eq:decision-milp}. Feasibility with $\varepsilon>0$ then confirms $\textrm{RC}(\phi)\ge d$. In the case of negative thresholds $c_\predicate<0$, $\log(c_\predicate)$ is not defined. Thm.~\ref{thm:predicate_shifting} 
below shows that shifting all thresholds by $\kappa > |\min_\predicate(c_\predicate)|$ and remapping signals by the same $\kappa$ preserves rank-realizability, reducing the problem to the positive-threshold case.

\begin{theorem}\label{thm:predicate_shifting}
Consider $\phi_\weightset$ in PNF with predicates $\mu := \signal^j \sim c_\predicate$, relations $\sim \in \{\ge,\le\}$, and $c_\predicate \in \mathbb{R}$. Let $\kappa$ be such that $c^{(\kappa)}_\predicate := c_\predicate + \kappa > 0$ for all predicates. Define $\phi^{(\kappa)}_\weightset$ by replacing each predicate threshold with $c^{(\kappa)}_\predicate$. If $\phi_\weightset$ is rank-realizable with respect to $\set{X} = \{\signal_1,\ldots,\signal_d\}$, then $\phi^{(\kappa)}_\weightset$ is rank-realizable with respect to $\set{X}^{(\kappa)} = \{\signal_i + \kappa\}_{i=1}^d$.
\end{theorem}

\begin{proof}
    Define $\signal^{(\kappa),j}_i = \signal^j_i + \kappa$. Then $\signal^{(\kappa),j} - c^{(\kappa)}_\predicate = \signal^j - c_\predicate$, therefore $\wrob(\signal,\phi_{\wval},t) = \wrob(\signal^{(\kappa)},\phi^{(\kappa)}_{\wval},t)$ for all weight valuations. Rank-realizability is therefore preserved under the $\kappa$-shift.
\end{proof}

To summarize rank-capacity approaches, when predicates are defined for different dimensions, we introduce variables for $\log(h_{\predicate}(s(t)))$ and solve the following problem to find whether there exists $d$ signals that the formula can rank in all combinations: 
\begin{corollary}\label{cor:rc_lb}
Consider $\phi_{\weightset}$. If the optimizer of the following problem is $\varepsilon^*>0$, then $\textrm{RC}(\phi) \geq d$:
\begin{equation}
    \label{eq:decision-corollary}
    \begin{array}{cc}
    \max & \varepsilon \\
    \suchthat &\eqref{eq:milp-flag1}-\eqref{eq:local-independence}, \eqref{eq:log-pred-upper-bound}, \eqref{eq:log-pred-lower-bound}\\
    & \signal_1(t), \ldots \signal_d(t) \in \sdomain_{>0}, \\ & \xi^i_{\phi, t}, \eta^i_{\phi, t}, \zeta^i_{\phi, t} \in \mathbb{B}, w_{\predicate, t} \in \reals^p, \\ & \epsilon^{i}_{\predicate, t} \ge 0, \varepsilon \geq 0 \end{array}
\end{equation}
\end{corollary}

Proof directly follows Thm.~\ref{thm:milp-correctness}. If the predicates share the same signal dimension, we apply the threshold shift of Thm.~\ref{thm:predicate_shifting}, use the $\log(s^j)-\log(c_{\predicate})$ and solve~\eqref{eq:decision-milp-capacity} for $d$ signals, which is still an MILP. If this problem is feasible with $\varepsilon^*>0$, then there exists a candidate set of signals. We need to certify candidates via~\eqref{eq:decision-milp}. If the rank-realizability problem is feasible with, again, a strictly positive value, then we have $\mathrm{RC}(\phi)\ge d$.

To find a tight lower bound, initialize at $d = \tilde{p}$ where $\tilde{p}$ is the number of predicate-time pairs, the maximum possible rank-capacity under our sufficient conditions, and decrease $d$ until~\eqref{eq:decision-corollary} is feasible with $\varepsilon^*>0$. Binary search over $d \in \{2,\ldots,\tilde{p}\}$ is a practical alternative.

\section{Experiments}
\subsection{Rank-realizability of Robotic Tasks}
In this experiment, we focus on a robot navigation task. Our goal is to determine whether a task specification can rank task-satisfying trajectories across all combinations. The robot starts at $(x, y)=(1,1)$, must visit region $A = [7,9]\times[1,3]$ or $B = [1,3]\times[7,9]$ within {the first $10$-second period}, then reach $C=[7,9]\times[7,9]$ within {the next $10$-second period} while staying in environment bounds $E = [0,10]\times[0,10]$ and avoiding unsafe region $U=[3,6]\times[3,6]$. This task is written as $\phi = \F_{[0,10]}((x,y) \in A \lor (x,y) \in B) \land \F_{[10,20]}\G((x,y) \in C) \land \G_{[0,20]}\lnot((x,y) \in U) \land \G_{[0,20]}((x,y) \in E)$. 

Fig.~\ref{fig:robots} shows twelve trajectories, each generated with random weight valuations using PyTelo~\cite{cardona2023flexible}. Let $\set{X}_r$ denote the set of trajectories. Solving the $\log$-transform of ~\eqref{eq:decision-milp} for $\phi$ and $\mathcal{X}_r$, yields an infeasible MILP, meaning $\phi$ does not satisfy the sufficient condition for rank-realizability over $\mathcal{X}_r$. Intuitively, this is expected because multiple trajectories differ very little in their positions over time. It is structurally difficult to place similar trajectories in all possible orderings.

\begin{figure}[hbt!]
    \centering
    \includegraphics[width=0.6\linewidth]{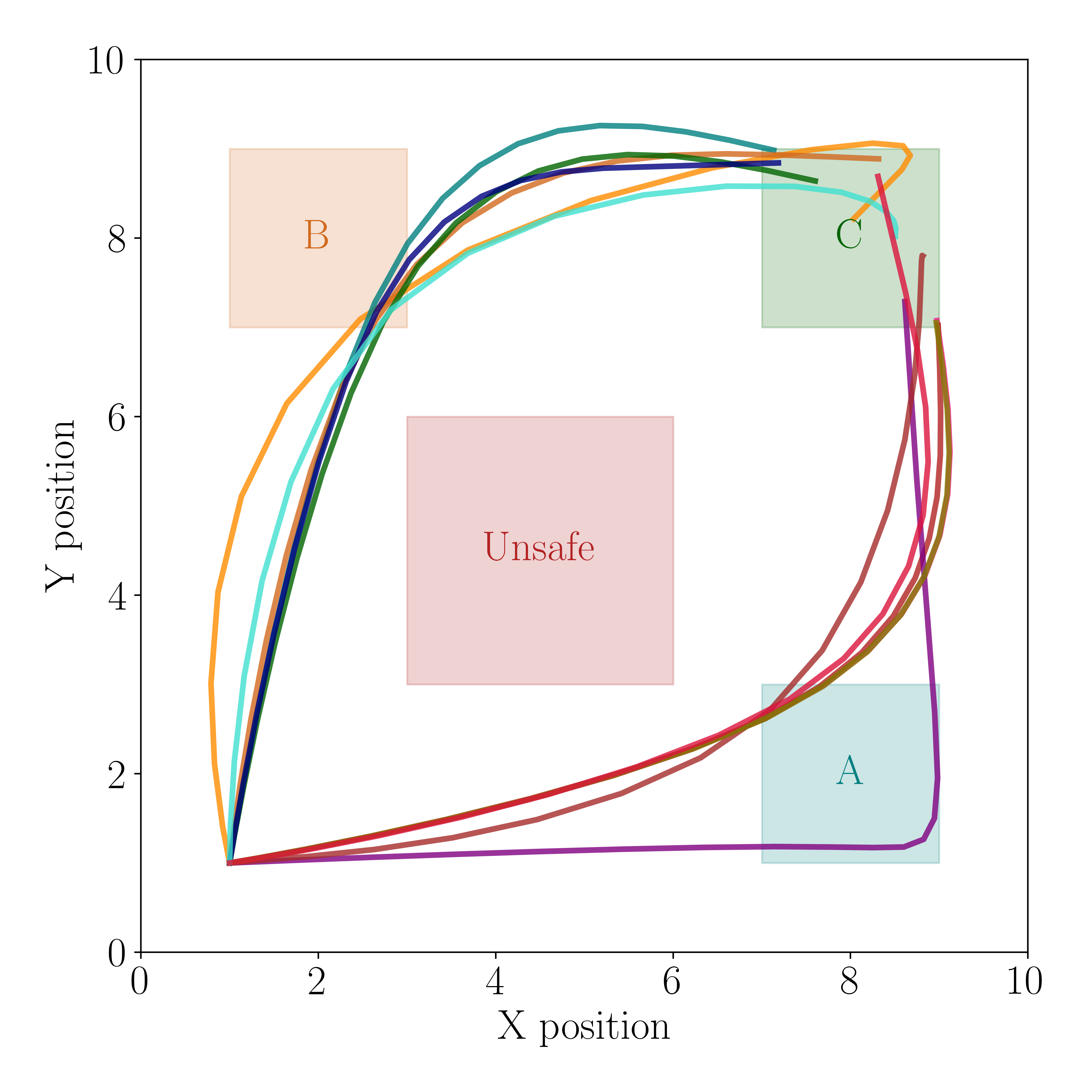}
    \caption{Trajectories satisfying the robot navigation task. Each trajectory starts at $(1,1)$ and reaches region $C$ in $T=20$ seconds, while visiting either region $A$ or region $B$. }
    \label{fig:robots}
\end{figure}

Solving the rank-capacity problem yields $\textrm{RC}(\phi) \ge 45$, which is substantially larger than  $|\set{X}_r|=12$. In fact, forty-five signals is the largest number of signals we experimented with and were able to return a response within a time limit of one hour, so $\textrm{RC}(\phi) \ge 45$ might be a loose lower bound. It is worth emphasizing that rank-capacity does not restrict signals from following system dynamics. So, even though there are at least $45$ trajectories in the entire domain over which $\phi$ is rank-realizable, the number of rank-realizable trajectories from a single system may be fewer than $45$.

\subsection{Rank-capacity Analysis of Boolean-Equivalent Formulas}

Motivated by Example~\ref{example:1}, we analyze how the rank-capacity changes among equivalent formulas in Boolean semantics. 
We consider $\phi= \G_{[0,T]}((\signal^1 \ge 0) \lor (\signal^2 \ge 0))$, and test whether it is possible to increase the rank-capacity of a formula by trivially appending it either by conjunction or disjunction. We construct $\phi_{(i)} = \phi \lor \phi$, $\phi_{(ii)} = \phi \land \phi$, $\phi_{(iii)} = \phi_{(i)} \land \phi_{(i)} $, and $\phi_{(iv)} = \phi \lor \phi \lor \phi \lor \phi$. For each formula, we find the tightest lower bound on the rank-capacity. We put a one-hour time limit on each MILP.

\begin{table}[hbt!]
\centering
\caption{RC of formulas that are Boolean-equivalent to $\phi$}
\addtolength{\tabcolsep}{5pt} 
    \begin{tabular}{ccccccc}
    \toprule
     $T$ & $\phi$ & $\phi_{(i)}$ & $\phi_{(ii)}$ &  $\phi_{(iii)}$ & $\phi_{(iv)}$  \\\midrule
    $1$ & $2$ & $4$ & $4$ & $8$ & $8$ & \\
    $2$ & $4$ & $8$ & $8$ & $16$ & $16$ &  \\
    $3$ & $6$ & $12$ & $12$ & $24$ & $24$  & \\
    \bottomrule
    \end{tabular}
    \label{tab:results}
\end{table}

Tab.~\ref{tab:results} shows that increasing the horizon of a formula will increase the rank-capacity, and that appending a formula with itself can increase the rank-capacity. However, we have a counterexample to this claim. Take $\varphi = (\signal^1 \ge 0) \lor (\signal^2 \ge 0) $, extension $\varphi_{(i)} = \varphi \lor \varphi$ will not increase the rank-capacity. Tab.~\ref{tab:boolean-results} shows results for appended versions of $\varphi$. 

\begin{table}[hbt!]
\centering
\caption{RC of formulas that are Boolean-equivalent to $\varphi$}
\addtolength{\tabcolsep}{-1.8pt} 
    \begin{tabular}{ccccccc}
    \toprule
       & $\varphi_{(i)}=$ & $\varphi_{(ii)}=$ & $\varphi_{(iii)}=$ & $\varphi_{(iv)}=$ & $\varphi_{(v)}=$\\ $\varphi$ & $\varphi \lor \varphi$ & $\varphi \land \varphi$ &  $\varphi_{(i)} \land \varphi_{(i)}$ & $\varphi_{(ii)} \land \varphi_{(ii)}$ & $\varphi_{(iii)} \lor \varphi_{(iii)}$ \\\midrule
     $2$ & $2$ & $4$ & $4$ & $8$ & $6$\\
    \bottomrule
    \end{tabular}
    \label{tab:boolean-results}
\end{table}

Tab.~\ref{tab:boolean-results} reveals that simply appending the formula with the disjunction or the conjunction does not necessarily increase its rank-capacity. Furthermore, alternating the outermost operator when appending does not guarantee the same gain compared to just repeating it (cf. $\varphi_{(ii)}$, $\varphi_{(iv)}$, and $\varphi_{(v)}$). 

Together, these results show that certain Boolean manipulations can increase the rank-capacity without changing the qualitative semantics of the formula, but when and how they do so remains unclear, underscoring the need for structural characterizations of rank-capacity beyond Boolean equivalence alone. 

\section{Conclusion, Limitations, and Future Work}
This work formalizes WSTL as a family of scoring functions over signals and addresses two foundational questions:  under what conditions a fixed WSTL template with adjustable weights can realize all permutations of a finite set of signals (rank‑realizability), and 
how this capability scales with the set size (rank‑capacity). The analysis is centered around base weights, which reveals structural properties, such as positive homogeneity and the influence of $\min/\max$ aggregations, clarifying how weights determine the resulting orderings. Building on this, we present an MILP to certify rank-realizability by finding a weight valuation that sets all \wstlrob values to one, assigns each signal a unique critical predicate-time pair, and separates all non-critical pairs by a common margin. The encoding captures the $\min/\max$ semantics of operators using indicator variables and the big-M method. A $\log$‑domain transformation maintains MILP tractability within the original weight space. For rank-capacity, we derive constructive lower bounds by treating signals as decision variables, clarifying that formula structure has a strong influence on capacity.


{Limitations include that the certificate is sufficient but not necessary. When the MILP is infeasible, one cannot distinguish between $\phi$ failing to be rank-realizable on $\set{X}$ and $\phi$ being rank-realizable on $\set{X}$ but failing the unique-distinct-critical-path condition. A complete but intractable alternative is to enumerate $|\set{X}|!$ rankings and solve the weight synthesis problem of \cite{karagulle2025sopl} for each, and construct $\langle \phi \rangle_{\set{X}}$. The proposed MILP trades completeness for tractability. Deriving similar, efficiently computable necessary conditions remains an open problem. The analysis assumes discrete-time signals; extension to continuous-time signals is an open direction. Finally, relaxing the strict ranking assumption to accommodate orderings with ties is a future research direction.}


\balance

\bibliographystyle{IEEEtran}
\bibliography{references}

\end{document}

%% file: macros.tex
\usepackage{cite}
\usepackage[pdftex]{graphicx}
\usepackage{amsmath}
\usepackage{amssymb}
\usepackage{amsfonts}
\usepackage{mathrsfs}
\usepackage{algpseudocode}
\usepackage{algorithm}
\usepackage{array}
\usepackage{xcolor}
\usepackage{tikz}
\usepackage[bb=dsserif]{mathalpha}
\usepackage{booktabs}
\usepackage{xspace}
\usepackage{balance}

\definecolor{tappanred}{HTML}{9A3324}
\definecolor{rossorange}{HTML}{D86018}
\definecolor{teal}{HTML}{008080}
\definecolor{crimson}{HTML}{DC143C}
\definecolor{darkorange}{HTML}{FF8C00}
\definecolor{green}{HTML}{A5A508}

\newtheorem{example}{Example}
\newtheorem{proposition}{Proposition}
\newtheorem{definition}{Definition}
\newtheorem{theorem}{Theorem}
\newtheorem{problem}{Problem}

\newtheorem{corollary}{Corollary}

\newcommand{\set}[1]{\mathcal{#1}}

\newcommand{\val}[1]{\boldsymbol{#1}}

\newcommand{\signal}{s}

\newcommand{\reals}{\mathbb{R}}
\newcommand{\integers}{\mathbb{Z}}

\newcommand{\sdomain}{\mathcal{S}}
\newcommand{\timedomain}{\mathbb{T}}

\newcommand{\mat}[1]{\begin{bmatrix} #1 \end{bmatrix}}
\newcommand{\suchthat}{\text{ s.t. }}

\newcommand{\range}[2]{{(#1{:}#2)}}

\newcommand{\wstlrob}{weighted robustness\xspace}

\newcommand{\G}{\square}
\newcommand{\F}{\Diamond}
\newcommand{\U}{\text{\bf{U}}}
\newcommand{\predicate}{\mu}

\newcommand{\lb}{a}
\newcommand{\ub}{b} 

\newcommand{\rob}{\rho}

\newcommand{\subf}{\sqsubseteq}
\newcommand{\fpath}{\mathfrak{p}}

\newcommand{\weight}{w}
\newcommand{\weightset}{\set{W}}
\newcommand{\wval}{\val{w}}

\newcommand{\wrob}{r}

\newcommand{\perm}{\sigma}
\newcommand{\permspace}{\Sigma}
\newcommand{\robspace}{P}

\algrenewcommand\algorithmicindent{1em}%
\let\oldReturn\Return
\renewcommand{\Return}{\State\oldReturn}

\usetikzlibrary{arrows.meta, fit, backgrounds}

\definecolor{defcol}{HTML}{D6EAF8}
\definecolor{thmcol}{HTML}{FDEBD0}
\definecolor{corcol}{HTML}{F9EBEA}
\definecolor{probcol}{HTML}{EDE7F6}
\definecolor{eqcol}{HTML}{FEF9E7}
\definecolor{toolcol}{HTML}{EAEDED}
\definecolor{seccol}{HTML}{FAFAFA}
\definecolor{defborder}{HTML}{2E86C1}
\definecolor{thmborder}{HTML}{CA6F1E}
\definecolor{corborder}{HTML}{C0392B}
\definecolor{probborder}{HTML}{6C3483}
\definecolor{eqborder}{HTML}{B7950B}
\definecolor{toolborder}{HTML}{7F8C8D}
\definecolor{corecol}{HTML}{E9F7EF}   
\definecolor{coreborder}{HTML}{27AE60}

%% file: flow_chart_iterations.tex
\begin{figure*}
\centering
\resizebox{\linewidth}{!}{%
\begin{tikzpicture}[
    core/.style ={draw=coreborder, fill=corecol, rounded corners=2pt, align=center, font=\small\bfseries, inner sep=1mm, line width=1.2pt, minimum height=0.6cm},
    def/.style  ={draw=defborder,  fill=defcol,  rounded corners=2pt, minimum width=1.8cm, align=center, font=\scriptsize, inner sep=1mm, line width=0.8pt, minimum height=0.8cm},
    thm/.style  ={draw=thmborder,  fill=thmcol,  rounded corners=2pt, align=center, font=\scriptsize, inner sep=1mm, line width=1.2pt, minimum height=0.8cm},
    cor/.style  ={draw=corborder,  fill=corcol,  rounded corners=2pt, align=center, font=\scriptsize, inner sep=1mm, line width=0.8pt, minimum height=0.8cm},
    eq/.style   ={draw=eqborder,   fill=eqcol,   rounded corners=2pt, align=center, font=\scriptsize, inner sep=1mm, line width=0.8pt, minimum height=0.8cm},
    tool/.style ={draw=toolborder, fill=toolcol, dashed, rounded corners=2pt, align=center, font=\scriptsize, inner sep=1mm, line width=0.8pt, minimum height=0.8cm},
    arr/.style  ={-{Stealth[length=5pt]}, line width=1.2pt, gray!55!black},
    seclabel/.style={font=\small\bfseries},
]

\def\xLeft{0.0}       
\def\xRCLeft{11.1}    
\def\xThmSix{13.3}
\def\xThmSeven{15.6}
\def\xRight{17.7}     

\def\yCore{0.0}       
\def\ySupport{-1.0}   
\def\yMain{-2.75}     
\def\yBottom{-4.2}   

\def\yBoxTop{-0.5}    
\def\yBoxBot{-4.4}    

\def\xCut{2.6}        
\def\yCutTop{-3.3}    

\node[core, minimum width=10.9cm, anchor=west] (Def1) at (\xLeft, \yCore)
    {\textbf{Def.~\ref{defn:ranking-set}:} Rank-Realizability};
\node[core, minimum width=6.5cm, anchor=west]  (Def2) at (\xRCLeft, \yCore)
    {\textbf{Def.~\ref{def:rank-capacity}:} Rank-Capacity};

\node[def, anchor=west]  (Def3) at (2.0, \ySupport)
    {\textbf{Def.~\ref{def:base-weights}:}\\ Base Weights};
\node[def, anchor=west]  (Def4) at (4.5, \ySupport)
    {\textbf{Def.~\ref{def:critical-path}:}\\ Critical Path};
\node[thm, minimum width=4cm, anchor=west]  (Thm1) at (6.75, \ySupport)
    {\textbf{Thm.~\ref{thm:indep-sufficient}:}\\ Sufficient Cond.\ for $|\langle\phi\rangle_{\set{X}}|\!=\!|\set{X}|!$};
\node[thm, minimum width=2cm, anchor=west]  (Thm6) at (\xThmSix, \ySupport)
    {\textbf{Thm.~\ref{thm:predicate-positivity}:}\\ Predicate Positivity};
\node[thm, minimum width=2cm, anchor=west] (Thm7) at (\xThmSeven, \ySupport)
    {\textbf{Thm.~\ref{thm:predicate_shifting}:}\\ Threshold Shift};

\node[tool, minimum width=2.5cm, anchor=north east] (StrPr) at (\xCut-0.1, \yCutTop-0.1)
    {$\log$-transform \&\\ Structural Pruning~\cite{karagulle2025sopl}};

\node[thm, minimum width=2cm, anchor=west]  (Thm2) at (\xLeft, \yMain)
    {\textbf{Thm.~\ref{thm:rank-real-equiv}:}\\ Pb.~\ref{prob:rank_realizability} $\implies$ Pb.~\ref{prob:ranking-set}};
\node[thm, minimum width=3.5cm, anchor=west]  (Thm3) at (3.65, \yMain)
    {\textbf{Thm.~\ref{thm:milp-correctness}:}\\ Pb.~\ref{prob:rank_realizability} $\iff$ Eq.~\eqref{eq:decision-milp}};
\node[thm, minimum width=3.5cm, anchor=west] (Thm5) at (\xRCLeft+0.3, \yMain)
    {\textbf{Thm.~\ref{thm:capacity-milp}:}\\ 
    Eq.~\eqref{eq:decision-milp-capacity} $\implies$ $\mathrm{RC}(\phi)\ge d$
    };

\node[thm, minimum width=2.5cm,  double, double distance=1.5pt, anchor=north west] (Thm4) at (8.05, \yCutTop-0.1)
    {\textbf{Thm.~\ref{thm:milp-equiv}:}\\ Eq.~\eqref{eq:decision-milp} $\iff$ [MILP]};
\node[cor, minimum width=2.5cm,  double, double distance=1.5pt, anchor=north west] (Cor2) at (14.25, \yCutTop-0.1)
    {\textbf{Cor.~\ref{cor:rc_lb}:}\\ 
    Eq.~\eqref{eq:decision-milp-capacity} $\implies$ [MILP]
    };

\node (figright) at (\xRight, 0) {};

\draw[arr] ([xshift=-4.45cm]Def1.south) -- (Thm2.north);
\draw[arr] (Def4.south)  -- (Thm3.north);
\draw[arr] (Thm2.east)   -- (Thm3.west);
\draw[arr] (Thm3.south)  |- (Thm4.west);
\draw[arr] (StrPr.east)  -- (Thm4.west);
\draw[arr] ([xshift=-1.6cm]Def2.south) -- ([xshift=-0.4cm]Thm5.north);
\draw[arr] ([xshift=-0.4cm]Thm5.south)  |- (Cor2.west);

\draw[arr] (Def3.south) -- ++(0,-0.3) -| (Thm3.north);


\draw[arr] (Thm1.south) -- ++(0,-0.3) -| (Thm3.north);

\draw[arr] (Thm6.south) -- ++(0,-0.3) -| (Cor2.north);

\draw[arr] (Thm7.south) -- ++(0,-0.3) -| (Cor2.north);

\draw[arr, -{Stealth[length=8pt]}, line width = 2pt, purple] (Thm3.east) -- (Thm5.west);
\draw[arr, -{Stealth[length=8pt]}, line width = 2pt, purple] (Thm4.east) -- (Cor2.west);
\begin{scope}[on background layer]
  \draw[draw=probborder!60, fill=probcol!20, rounded corners=4pt]
    (-0.1,    \yBoxTop) --
    (10.9,    \yBoxTop) --
    (10.9,    \yBoxBot) --
    (\xCut,   \yBoxBot)
    node[midway, below, font=\small\bfseries, probborder] {Rank-Realizability}
    -- (\xCut,   \yCutTop) --
    (-0.1,    \yCutTop) --
    cycle;
  \draw[draw=probborder!60, fill=probcol!20, rounded corners=4pt]
    (\xRCLeft, \yBoxTop) --
    (\xRight,  \yBoxTop) --
    (\xRight,  \yBoxBot) --
    (\xRCLeft, \yBoxBot)
    node[midway, below, font=\small\bfseries, probborder] {Rank-Capacity}
    -- cycle;
  \draw[draw=coreborder, fill=corecol!50, rounded corners=4pt]  
    (-0.1, 0.4)--
    node[midway, above, font=\small\bfseries, coreborder] {Core Problems}
    (\xRight,  0.4) --
    (\xRight,  -0.4) --
    (-0.1, -0.4)
    -- cycle;
\end{scope}

\end{tikzpicture}%
}
\caption{Summary and dependencies of theoretical results. The green band marks the two core problems, each anchoring its respective section; blue nodes represent intermediate definitions, orange nodes represent theorems, and red nodes represent corollaries. Arrows indicate logical dependencies. The crimson arrow highlights the key cross-section dependency, and double borders mark the final results of each section.}
\label{fig:summary-chart}
\end{figure*}